%% file: cm_Industrial_Informatics.tex
\documentclass[journal]{IEEEtran}
\usepackage{cite}
\usepackage{amsmath}
\usepackage{url}
\usepackage{algorithm}
\ifCLASSOPTIONcompsoc
\usepackage[caption=false,font=normalsize,labelfont=sf,textfont=sf]{subfig}
\else
\usepackage[caption=false,font=footnotesize]{subfig}
\fi
\usepackage{algpseudocode}
\usepackage{float}
\input{mymath}
\begin{document}
\title{Application of Correlation Indices on Intrusion Detection Systems: Protecting the Power Grid Against Coordinated Attacks}
\author{Christian~Moya,~\IEEEmembership{Student~Member,~IEEE,}
		~Junho~Hong,~\IEEEmembership{Member,~IEEE,}        
        ~and~Jiankang~Wang,~\IEEEmembership{Member,~IEEE}
\thanks{C. Moya and J.K. Wang are with the Department of Electrical and Computer Engineering, The Ohio State University, Columbus, OH, 43210 USA e-mail: \{moyacalderon.1,wang.6536\}@osu.edu.}
\thanks{J. Hong is with the Energy and Automation Department, ABB US Corporate Research Center, Raleigh, NC, 27519 USA e-mail: junho.hong@us.abb.com.}}

\maketitle
\begin{abstract}
The future power grid will be characterized by the pervasive use of heterogeneous and non-proprietary information and communication technology, which exposes the power grid to a broad scope of cyber-attacks. In particular, Monitoring-Control Attacks~(MCA) --\ie attacks in which adversaries manipulate control decisions by fabricating measurement signals in the feedback loop-- are highly threatening. This is because, MCAs are (i)~more likely to happen with greater attack surface and lower cost,~(ii) difficult to detect by hiding in measurement signals, and (iii)~capable of inflicting severe consequences by coordinating attack resources. To defend against MCAs, we have developed a semantic analysis framework for Intrusion Detection Systems~(IDS) in power grids. The framework consists of two parts running in parallel: a Correlation Index Generator~(CIG), which indexes correlated MCAs, and a Correlation Knowledge-Base~(CKB), which is updated aperiodically with attacks' Correlation Indices~(CI). The framework has the advantage of detecting MCAs and estimating attack consequences with promising runtime and detection accuracy. To evaluate the performance of the framework, we computed its false alarm rates under different attack scenarios.     
\end{abstract}
\IEEEpeerreviewmaketitle
\input{Introduction}
\input{Background}
\input{MathematicalModels}
\input{CIG}
\input{CKB}
\input{Experiments}
\section{Conclusion}
In this paper, we developed a semantic analysis framework for Intrusion Detection Systems~(IDS) against Monitor-Control Attacks (MCA) in power grids. The framework has two parts running in parallel with IDS: A Correlation Index Generator~(CIG) that analyzes the correlation of potential hostile behaviors and indexes these behaviors, and a Correlation Knowledge-Base~(CKB) that is updated with the Indices generated by CIG. The performance of the proposed framework is evaluated under different attack scenarios in a cyber-physical setting. It is shown that the proposed framework is capable of detecting MCA and estimating attack consequences with promising runtime and detection accuracy. In addition, the experiments show that the detection outcome of the proposed framework is sensitive to both the size and locations of attack goals. Future work includes developing methods, which adapt CKB parameter settings to attack activities, to achieve an optimal trade-off between the FNR/FPR and detection runtime.   
\appendices
\section{Proof of Proposition~\ref{prop:1}}
\begin{proof}
Suppose, to get a contradiction, that $S':=\WC \setminus \{s_k^*\}$ is an effective MCA. Thus, $S' \subset \WC$ is a correlated MCA with cardinality $|S'|< |\WC|=: \kappa^*$, which contradicts the fact that $\WC$ is a strongly correlated MCA, \ie a CI with minimum cardinality. This proves the proposition.
\end{proof}
\input{PseudoCode}
\bibliographystyle{IEEEtran}
\bibliography{bib/cm_ref}
\end{document}

%% file: mymath.tex
\usepackage{amssymb,amsmath,amsthm,graphicx,nomencl,xpatch}
\usepackage{color}
\theoremstyle{plain}

\newtheorem{proposition}{Proposition}

\newtheorem{definition}{Definition}

\newtheorem{remark}{Remark}

\newlength{\nomitemorigsep}
\xpatchcmd{\thenomenclature}{%
  \section*{\nomname}
}{\vspace{-1em}
}{\typeout{Success}}{\typeout{Failure}}

\renewcommand*{\nompreamble}{\begin{multicols}{2}}
\renewcommand*{\nompostamble}{\end{multicols}}

\newcommand{\R}{{\mathbb R}}

\def\eg{\emph{e.g., }}
\def\ie{\emph{i.e., }} 
\def\dist{0.7em}

\DeclareMathOperator*{\diag}{diag}

\newcommand{\WC}{{S_{\alpha,j}}}       
\newcommand{\SC}{{S_{\alpha,j}^*}}       
\newcommand{\Con}{{G_{\alpha,j}}}      

%% file: Introduction.tex
\section{Introduction}
\IEEEPARstart{T}{he} power grid is evolving with increasing dependency on Information and Communication Technologies~(ICT). Today, ICT is realized in energy control centers through Supervisory Control and Data Acquisition~(SCADA) systems and Energy Managment Systems~(EMS). While EMSs make commands for power grid operation, SCADA systems serve as the gateway between EMS and field networks by passing measurements and control commands. The present SCADA is in the fourth generation of architectures, which bring innovative and cost-efficient solutions, such as cloud computing and Internet of Things, while opening up a much wider scope of cyber-security concerns among utilities~\cite{zhu2010scada}. Since the notorious Stuxnet attack to Siemens SIMATIC WinCC SCADA system in July 2010, approximately 45,000 cases of SCADA infection around the world have been reported, including the Iranian nuclear facilities and the Ukrainian power grid, according to Symantec's statistics~\cite{lee2016analysis}. These attacks, if successful, would lead to massive power outages, resulting in severe physical, economic, and social impacts.       

Intrusion Detection Systems~(IDS) are redeemed critical to protecting SCADA from cyber-attacks. In contrast to those methods aiming at strengthening the perimeter surrounding SCADA, IDSs generate `burglar alarms' whenever the security of the system is compromised \cite{hasan2016optimal}. To increase the chances of mounting a successful defense, the Department of Homeland Security recommends a combination of firewalls, De-militarized Zones, and IDSs grounded on the principle of defense-in-depth \cite{kuipers2006control}. 

While IDSs for traditional ICT systems are mature, implementing IDS in industrial control systems, such as power grids' SCADA, is facing unprecedented challenges in twofold. First, the power grid is a cyber-physical system, wherein continuity of operation is critical. Unlike traditional ICT systems, in which the effects of false alarms are limited to computer operations, false alarms in power grids would disrupt dependent vital physical processes and inflict severe consequences. Therefore, false positive (which falsely generates alarms for normal actions) is unacceptable whereas low false negative rate is desired. Second, the power grid is a real-time dynamical system. Any delay of control actions could lead to instabilities from local plant angle instability to inter-area oscillation~\cite{wang2017analysis}. In the extremity, delayed response of protective devices will cause cascading blackouts over a large scale. For this reason, propagation latency of control and measurement signals induced from IDS audit and process must be minimized. 

To address the first challenge, recent works develop IDS by integrating contextual information of power grids\cite{lin2016runtime,pan2015developing,vellaithurai2015cpindex,sridhar2014model,sun2016coordinated,liu2010security,vrakopoulou2015cyber}. The most common approach is to identify attacks based on their impact on power grids. For example, in \cite{vrakopoulou2015cyber}, Bayesian network models for the whole cyber-infrastructure and underlying power grids are constructed based on SCADA logs along with power network topological information. Power contingencies are then simulated on the Bayesian model to rank the severity of a detected cyber-intrusion. In \cite{sun2016coordinated,ten2008vulnerability,lin2016runtime}, IDSs audit and select packets that contain control commands, which (dis)connect grid components, \eg generators, transmission lines and substations. Cyber-attacks are identified if the power flow diverges in simulation under those control commands.  

Another approach is to calibrate the detection results in cyber-space with historical data of power grid operation, wherein data mining techniques are often applied. For example, deviations between current and historical Area Control Errors are used as indicators of cyber-attacks to Automatic Generation Control in EMS~\cite{sridhar2014model}. A hybrid IDS is developed in \cite{pan2015developing} that learns temporal state-based specifications for power grid scenarios of physical disturbances, cyber-attacks, and normal operations. 

However, both of these approaches share the common deficiency of requiring a long runtime, exacerbating the second challenge. While the former approach simulates power grids' response, which is a non-trivial task given the enormous size of power networks and the number of grid devices, the latter approach relies on frequent auditing and processing historical data over a sufficiently long period in order to ensure the desired accuracy. These put a high requirement on IDS accounting resources and could significantly reduce IDSs' performance in timely processing and propagating the information to grid functions and responsible defense authorities. 

Despite initial attempts on reducing IDS runtime in~\cite{lin2016runtime,ten2016cyber}, they are restricted to certain attack groups, wherein attacks are aimed at individual grid components and assume a single step in the cyber-physical causal chain (\ie adversaries directly disconnect grid devices through remote control); they are not able to handle more sophisticated attacks that are coordinated and through EMS. These attacks are defined as \emph{Monitoring-Control Attacks (MCA)} and considered highly threatening \cite{song2017smart}, because they are (i) more likely to happen with greater attack surface and lower attack cost, (ii)~difficult to detect by hiding in measurement signals and masquerading through EMS, and (iii) capable of inflicting much more severe consequences at a greater scale by coordinating attack resources targeting at multiple grid components. Although MCAs' attack mechanisms and physical impacts have been studied in a few works \cite{xie2010false,liang2016vulnerability,wang2016attack,jk0216attack}, there is no effective IDS solution available to defend against MCAs.    

To bridge this gap, this paper presents a semantic analysis framework for IDSs in power grids, which detects MCAs with promising runtime and detection accuracy. The framework is implemented as two parts running in parallel in IDS: a Correlation Index Generator (CIG), which indexes correlated attacks, and a Correlation Knowledge-Base (CKB), which is updated aperiodically with attacks' Correlation Indices~(CI). In addition, this paper makes the following contribution:
\begin{itemize}
\item A theoretical basis for CIG. We formulate MCAs as a bi-level mix-integer optimization program and solve it to provide CI solutions.
\item A suite of detection rules for CKB. Derived from set theory, these rules characterize the relation between adversaries' goals and coordinated attacks, thus enabling CKB to detect MCAs at runtime.
\item Defense strategies against MCAs. While most IDSs are passive, that is, they only generate ``burglar alarms'', our proposed method actively derives defense strategies against MCAs using a set-theoretic approach.    
\end{itemize}

The rest of the paper is organized as follows. Section II introduces the threat model, MCAs mechanisms and IDS implementation of the proposed semantics framework. Section III presents the mathematical model of power grids and MCAs. The theoretic basis for CIG and detection rules for CKB are derived in Section IV and V. In Section VI, the performance of proposed semantic framework is demonstrated with numerical experiments. Finally, all results of this paper are concluded in Section VII. While the proposed framework is capable of defending against less sophisticated attacks, such as control attacks, we elaborate the framework's working principle mainly based on MCAs in this paper.    
\vspace{-\dist}

%% file: Background.tex
\section{Background}
In order to develop the semantic analysis framework for IDSs in power grids, we need to consider three factors: the environment in which intrusions occur (the threat model), the intrusions we wish to detect (MCAs), and the intrusion detector (IDS implementation).
\vspace{-\dist}
\subsection{Threat Model}
\noindent
In the previous generations, SCADA activities were basically confined to proprietary networks. In contrast, the current fourth generation of SCADA is mostly internet-based, as illustrated in Fig.~\ref{fig:SCADA}. In particular, a large amount of measurement signals from transducers of grid equipment (\eg relays, generators and switch gears) are transmitted with raw data protocol in field networks \cite{hasan2016optimal}. This widens the cyber-attack surface in the following attack entry points as numbered in Fig.~\ref{fig:SCADA}~\cite{zhu2010scada}: 
\begin{enumerate}
\item[(1)] Directly hack into field devices, including transducers, actuators and meters. 
\item[(2)] Attack field network links between devices and from devices to Energy Control Centers~(ECC).
\item[(3)] Attack from inside of the ECC. This could happen within or external of the security enclaves, which boundaries are defined by the trust nodes (\eg firewall and IDS) \cite{yang2014multiattribute}. 
\item[(4)] Attack from inside enterprises functions or attack at its perimeter networks. 
\end{enumerate} 

\begin{figure}[t]
\centering
\includegraphics[width= .45\textwidth, height = 4.1 in]{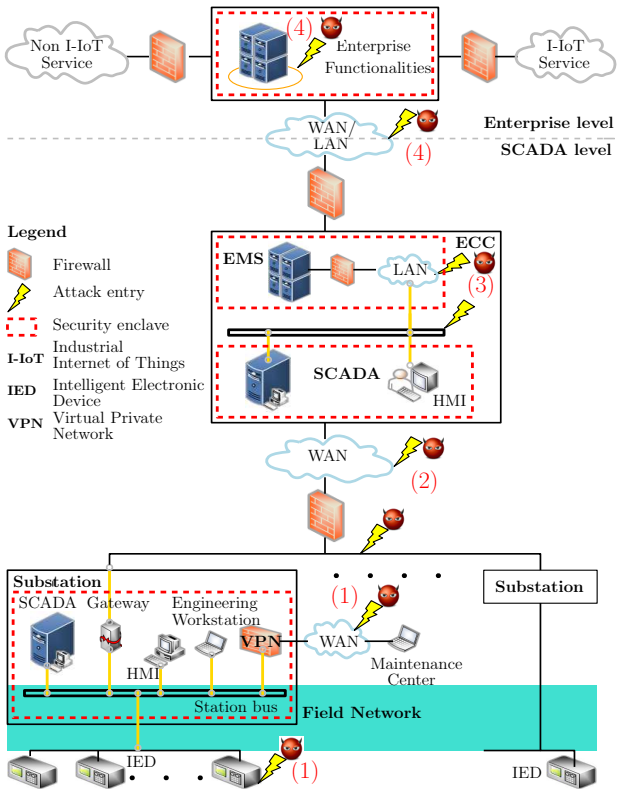}
\caption{EMS/SCADA Power Grid and Attack Entry Points.}
\label{fig:SCADA}
\end{figure} 

Through these chanels, adversaries can install malware, sniff, inject and modify host files and network traffic \cite{liu2009cybersecurity,zhu2010scada,ten2007vulnerability}. Based on the above fact, we make the following assumptions about the threat model: 
\begin{enumerate}
\item Adversaries can remotely penetrate the Local Area Network (LAN) and Wide Area Network (WAN). Though insider attacks outside security enclaves are allowed under the proposed framework, it is not our focus. We do not consider insider attacks within the security enclaves.
\item In ECC, we trust EMS. In other words, attacks are only executed on packets containing control and measurement signals that are transmitted over the network; they do not damage the EMS functions nor alter its encoded working principles.
\item IDSs are secure (\ie not compromised). In addition, we assume there are separate computing machines dedicated to IDSs that implement the proposed semantic analysis framework. Therefore, IDSs do not introduce extra vulnerabilities into power grids. 
\item IDS communication is secure. In other words, IDSs can safely exchange data.   
\item We do not consider attacks through enterprises functions. Launching MCAs through this path, though theoretically possible, is much more likely to fail due to extra layers of trust nodes.
\end{enumerate}
\vspace{-\dist}
\subsection{Monitoring-Control Attacks}
There are two clases of attack mechanisms in power grids, control attacks and monitoring attacks \cite{song2017smart}. They are illustrated in with a generic control diagram in Fig.~\ref{fig:attacks}. Control attacks refer to attacks that directly hijack and falsify control commands in power grids, such as disconnecting transmission lines and changing the power output of generators \cite{liu2009cybersecurity,lin2016runtime,ten2016cyber}. While able to inflict immediate physical consequences, they are less likely to occur in practice due to the restricted communication channels and easiness of detection. For example in conventional substations, relay commands, which trigger circuit breakers, are usually transmitted over proprietary communication channels or hard wire connection; generator power adjustments are requested through Human Machine Interface~(HMI), where operators would block and report suspicious actions. 
\begin{figure}[t]
    \centering 
    \subfloat[][Measurement attack]{\includegraphics[width=0.25\textwidth, height=1.05in]{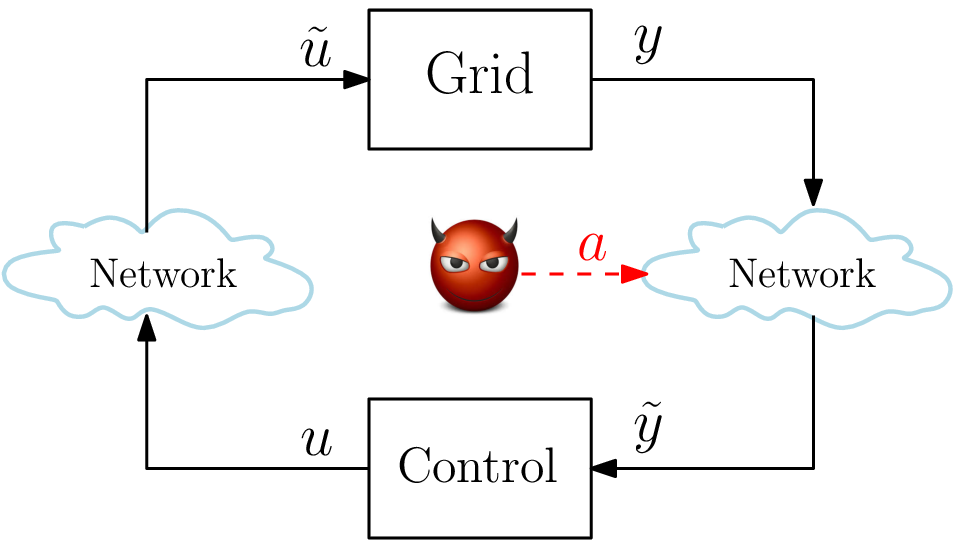}\label{fig:mattk}}
    ~
    \subfloat[][Control attack]{\includegraphics[width=0.25\textwidth, height=1.05in]{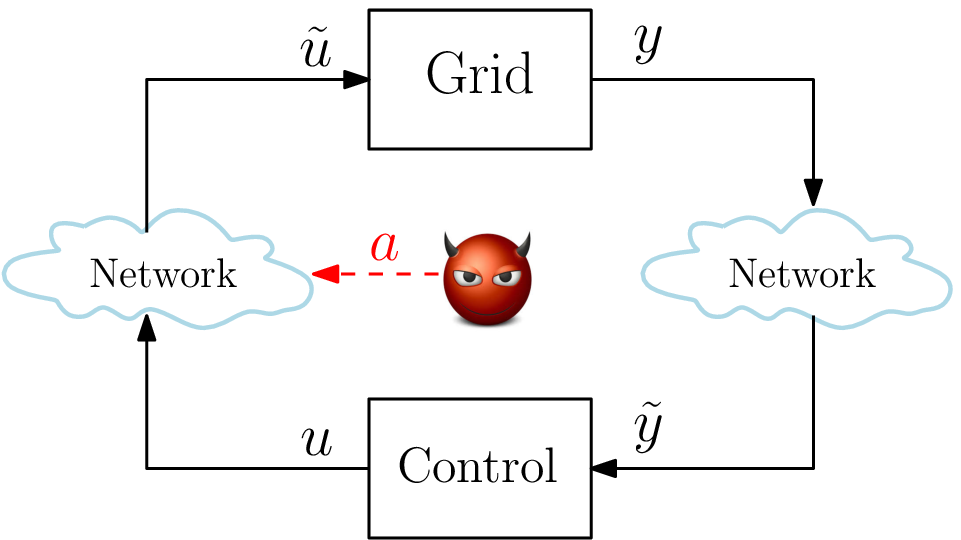}\label{fig:ctrl}}
    \caption{Control and measurement attacks. $u:$ control command, $y:$ measurements. $u \neq \tilde{u}$ ($y \neq \tilde{y}$) during the cyber-attack, where $\tilde{u}$ and $\tilde{y}$ are the corrupted control and measurement signals.}
    \vspace{-\dist}
    \label{fig:attacks}
\end{figure}

Monitoring attacks contaminate or eavesdrop measurements collected from transducers. In contrast to control commands, measurement signals have been more often transmitted over open-communication channels (\ie without any available authentication method) due to their large transmission volume and high transmission rate. This opens a wider cyber-surface to attacks. An important subset of monitoring attacks is \emph{Monitoring-Control Attacks}, in which adversaries manipulate control decisions by fabricating measurement signals in the feedback loop. On one hand, MCAs are difficult to detect, since the attack goals are hidden behind measurements and the control mechanisms. Thus, they cannot be inspected and intervened by human operators. On the other hand, they can inflict severe consequence by coordinating attack resources targeting at many measurements simultaneously; they are different from non-disruptive monitoring attacks that only exploit private information. Therefore, MCAs are considered highly threatening.    

MCAs' mechanism in power grids is briefed next. Main control functions of the power grid are realized through EMS, which consists of four blocks: network model-building (including topology processor and state estimation), security assessment, automatic generation control, and dispatch. Information flows within EMS are shown in Fig.~\ref{fig:AttackPaths}. In path 1, contaminated measurements drive control decisions in automatic generation control and dispatch after going through network-building models. While state estimation could effectively correct and identify bad data, a rich body of literature has demonstrated that contaminated measurements can still be injected through when the measurement errors are within the tolerance and/or the measurements are structure-wise conforming \cite{liu2011false,liufalse,liang2014cyber}. In path 2, contaminated measurements directly drive control decisions, as it is common for system operators to make a decision based on raw measurements in security constrained dispatch. Through both paths, adversaries may realize goals, such as depriving profit in electricity markets, disturbing power grid frequency and overloading grid equipment, causing tremendous financial losses, sabotaging, or even interrupting continuous grid operation.   

\begin{figure}[t]
\centering
\includegraphics[width= .40\textwidth, height = 1.70 in]{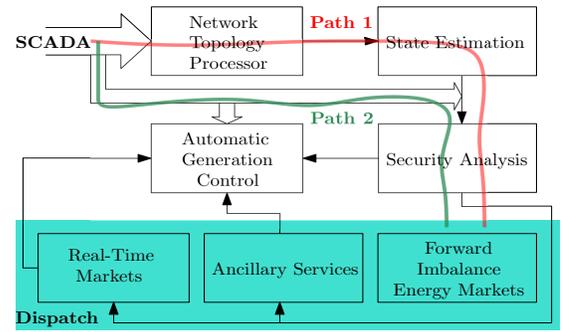}
\caption{Attack Paths on~EMS. In path~1, the adversary go through State Estimation and its screening methods. In path~2, the adversary can inject any attack signal that deceives the operator.}
\label{fig:AttackPaths}
\end{figure}
\vspace{-\dist}
\subsection{IDS Implementation}
\subsubsection{Proposed Framework in IDS Architecture} A general IDS architecture is defined with four modules, Event (E-blocks), Analysis (A-blocks), Database (D-blocks), and Response (R-blocks), as shown in Fig.~\ref{fig:IDSCIGCKB}~\cite{ord1998common}. The proposed semantic analysis framework has two parts: Correlation Index Generator (CIG) and Correlation Knowledge Base (CKB). They are aimed to provide contextual information of power grids additional to the traits that IDS sensed in the cyber-space (\eg host syslog and network traffic). 

\begin{figure}[t]
\centering
\includegraphics[width= .38\textwidth, height = 1.25 in]{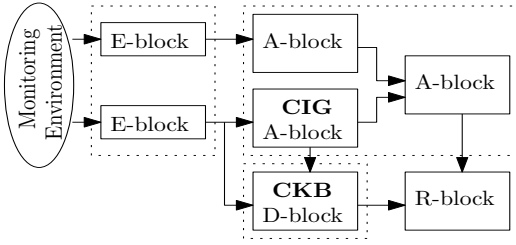}
\caption{Proposed IDS working principles~\cite{ord1998common}. E-blocks are Event blocks/IDSs' sensors, A-blocks are analysis blocks, D-blocks are database blocks, and R-blocks are response or mitigation blocks. In particular, CIG is an A-block and CKB is a D-block.}
\label{fig:IDSCIGCKB}
\end{figure}

CIG, depicted in Fig.~\ref{fig:CIG}, belongs to A-blocks. It analyzes the correlation of the potential hostile behaviors sensed by E-blocks, and indexes these behaviors with inductive-deductive patterns. For example, if a set of measurements are suspected to be contaminated, CIG first induces their consequence on the power grid with optimal power flow. If a transmission line is overloaded, then these measurements are \emph{weakly correlated}. Next, CIG deduces the critical measurements required to overload the transmission line. These critical measurements are \emph{strongly correlated} and will be represented by a set of Correlation Indices~(CI). The inductive-deductive patterns ensure minimal false negative rates that might be caused by normal deviations, such as noises and faults. In addition, CIG can be used to protect critical grid assets from MCAs, in which case CIs can be directly deduced from the predicted failures of these assets. Details about CIG are provided in Section~IV.

\begin{figure}[t]
\centering
\includegraphics[width= .40\textwidth, height = 2.20 in]{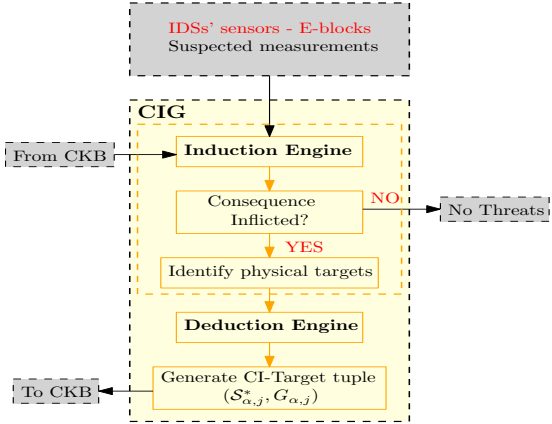}
\caption{Correlation Index Generator~(CIG).}
\label{fig:CIG}
\end{figure}

CKB, depicted in Fig.~\ref{fig:CKB}, belongs to D-blocks. It is updated with the CIs generated from CIG at an adaptive rate, which is determined by (i) configuration change of power networks, (ii) power grid stress level, (iii) detection rate of potential hostile events of E-blocks, and (iv) human operator's settings. At runtime, measurements detected by E-blocks are compared with the CIs in CKB. If the comparison is positive, then these measurements are considered forming an MCA. This information is passed to other A-blocks and R-blocks for further response. Since CKB does not contain any computation function, apart from arithmetic operation for CI comparison, it allows fast contextual information integration in IDSs. Details about CKB are provided in Section~V.

\begin{figure}[t]
\centering
\includegraphics[width= .43\textwidth, height = 2.20 in]{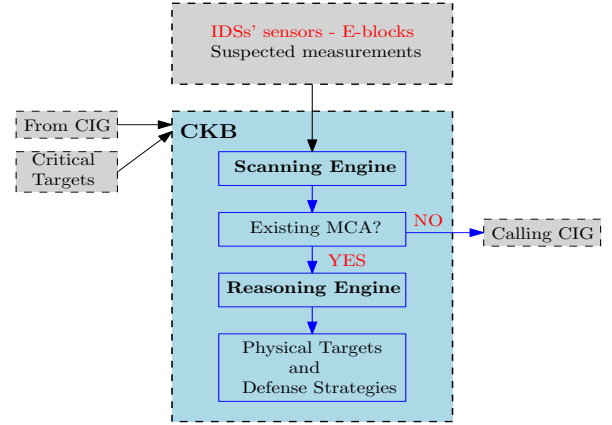}
\caption{Correlation Knowledge Base~(CKB).}
\label{fig:CKB}
\vspace{-\dist}
\end{figure}

Derived from set theory, defense strategies are proposed for R-blocks. The design of E-blocks is out of the scope of this paper. 
\subsubsection{IDS Dimensions} We consider two dimensions of IDS implementation related to the proposed semantic analysis framework. The proposed framework is flexible in implementation in the other dimensions, such as audit source (\ie host- or network-based detection), audit frequency and continuity, which definitions are given in the survey \cite{debar1999towards}. 
            
\noindent
\textbf{Detection Approach.} There are two main detection approaches in IDS development: signature- and anomaly-based. In between these approaches lie the probabilistic- and specification-based methods \cite{Axelsson2000,zhu2010scada,debar1999towards}. All of these approaches are based on \emph{direct} knowledge of cyber-activities (\ie host syslog and network traffic). In complementary, behavioral detection approaches capture the patterns, which are not necessarily illegitimate in a direct setting but wrong in a \emph{contextual} setting as a secondary evidence. The proposed analysis framework belongs to the last class and will be implemented with other direct knowledge-based approaches in IDS. 

\noindent
\textbf{Distributed v.s. Centralized.} The proposed analysis framework can be implemented under centralized, distributed or hierarchical structure of IDS. Provided the cost and communication constraints in power grids, we consider IDSs are only installed at the substation level and above, but not at individual Intelligent Electronic Devices or Remote Terminal Units. Thus, under a centralized structure, the proposed framework will allow IDS at a substation to detect and identify MCAs within its service area. For MCAs across service areas under multiple substations, a distributed structure is needed, wherein IDS at substations have peer-to-peer communication so that detected events can be exchanged. Alternatively, a hierarchical structure can be formed. The proposed analysis framework is integrated at a master IDS, which supervises all the substation IDSs by collecting, analyzing their detected events and sending instructions for detected MCAs.  
\vspace{-\dist}

%% file: MathematicalModels.tex
\section{Mathematical Models}
In this section, we model the power grid, dispatch applications, and Monitoring Control Attacks. 
\vspace{-\dist}
\subsection{Mathematical Notation}
Throughout this paper, we use the following mathematical notation. Let $\R$ and $\R_{\geq 0}$ (resp. $\R_{>0}$) denote the set of real numbers and the set of non-negative (resp. positive) real numbers. We let $\mathbf{1}$ and $\mathbf{0}$ denote, respectively, the vectors or matrices with all components equal to one and zero. Given a finite set $V$, we let $|V|$ denote its cardinality, \ie the number of elements of $V$, and $2^V$ the power set of $V$, \ie the set of all subsets of $V$.

For a matrix $A \in \R^{n \times m}$, we let $[A]_i$ denote its $i$th row. For a vector $x \in \R^n$, $x_i$ denotes its $i$th element, $\diag(x)$ the diagonal matrix of $x$, and $||x||_0$ the zero-norm of $x$, \ie the number of non-zero elements of $x$. 
\vspace{-\dist}
\subsection{Power Grid Model}
We model the power grid as the graph~$G = (V,E)$, where $V$ is the set of $n$ buses and $E \subset V \times V$ is the set of $m$ transmission lines. To each bus~$i \in V$, we associate the demand (or consumption) $P_{d,i} \in \R_{\geq 0}$. In addition, let $V_g \subset V$ denote the set of $n_g$ buses with dispatchable generation. To each generator bus $i \in V_g$, we associate the power generated~$P_{g,i} \in \R_{>0}$. Similarly, to each transmission line~$l:=(i,j) \in E$ connecting buses $i$ and $j$, we associate the power flow $P_{f,l} \in \R$. In vector form, the demand, generation, and power flows are, respectively, $P_d = [P_{d,1},P_{d,2},\hdots,P_{d,n}]^\top$, $P_g = [P_{g,1},P_{g,2},\hdots,P_{g,n_g}]^\top$, and $P_f = [P_{f,1},P_{f,2},\hdots,P_{d,m}]^\top$. 

The power grid is assumed to have a set of $n_s$ substations, \ie $S: = \{s_1,s_2,\hdots,s_{n_s}\}$. We model the power grid within substation~$s_k$'s service area as the sub-graph $G_{s_k} = (V_{s_k},E_{s_k})$ with the following properties:
\begin{enumerate}
\item All substations' service areas compose the power grid, \ie $G = \cup_{s_k \in S} G_{s_k}$.
\item Substations' service areas might overlap, \ie for some $s_j,s_k \in S$, we may have $G_{s_j} \cap G_{s_k} \neq \emptyset$.
\item The overlapped areas do not contain generator buses.
\item Each substation~$s_k$ collects demand measurements, denoted as $\tilde{P}_d \in \R^n_{\geq 0}$, within its service area, \ie all $\tilde{P}_{d,i}$ such that $i \in V_{s_k}$.
\end{enumerate}
\vspace{-\dist}
\subsection{Dispatch Application Model}
\textit{Dispatch applications} in EMS compute the generation output for the grid, denoted as $P_g^+ \in \R^{n_g}_{>0}$, by observing demand measurements~$\tilde{P}_d$ and using security constrained optimal power flows. These applications are triggered based on a guard condition (\ie a boolean condition). This guard condition is enabled by a security assessment algorithm (which usually involves network model-building), or by a system operator during real-time and contingency dispatch. Examples include generation dispatch in Real-Time Markets and Ancillary Services (see Fig.~\ref{fig:AttackPaths}).

Dispatch applications are based on the active and reactive power flow model, which describes how power balances on buses and flows on transmission lines. However, computing this coupled power flow may become computationally intractable for large-scale power grids. For this reason, the decoupled DC power flow is commonly adopted by operators when the power grid is in the normal status~\cite{dorfler2013novel}. The linearity and sparsity in the DC power flow allows much faster computation.

We formulate the security constrained DC optimal power flow as a convex optimization problem that minimizes the generation cost~\eqref{eq:SecOPFCost}, balances generation and demand~\eqref{eq:SecOPFBalance}, and keeps the generation~\eqref{eq:SecOPFMaxGen} and power flows~\eqref{eq:SecOPFMaxFlow} within operational limits, \ie
\begin{subequations} \label{eq:SecOPF}
\begin{align}
\Omega(\tilde{P}_d):~\min_{P_g} \quad & \frac{1}{2} P_g^\top C_2 P_g + c_1^\top P_g +c_0, \label{eq:SecOPFCost} \\
					 \text{s.t.} \quad & \mathbf{1}^\top P_g - \mathbf{1}^\top \tilde{P}_d = 0, \label{eq:SecOPFBalance}\\
                  & P_g \in [\mathbf{0},\bar{P}_g], \label{eq:SecOPFMaxGen}\\
                  & \underbrace{F(\Pi_g P_g - \tilde{P}_d)}_{=: P_f} \in [-\bar{P}_f,\bar{P}_f], \label{eq:SecOPFMaxFlow} 
\end{align}
\end{subequations}
where $c_0,c_1,c_2 \in \R^n_{\geq 0}$ are the cost coefficients for generators, $C_2 = \diag (c_2)$, $\bar{P}_g \in \R^n_{\geq 0}$ is the rated power from generators, $\bar{P}_f \in \R^m_{\geq 0}$ is the thermal capacity of transmission lines, $F \in \R^{m \times n}$ is the generator shift matrix, and $\Pi_g \in \{0,1\}^{n \times n_g}$ is a matrix that maps generator buses to buses.   

Thus, given the demand measurements~$\tilde{P}_d$, an optimal solution $P_g^+ \in \Omega(\tilde{P}_d)$ corresponds to the new generation output for the grid.  
\vspace{-\dist}
\subsection{Attack Model}
In this subsection, we define MCAs, attack goals, and attack constraints. We also describe two types of MCAs: strongly and weakly correlated. 
\subsubsection*{Monitoring Control Attacks} MCAs aim to manipulate dispatch applications in EMS. In an MCA, adversaries hack into substations' ICT. The corrupted measurements are modeled as follows:
\begin{align} \label{eq:MeasAttacks}
\tilde{P}_d(a) = P_d +a, 
\end{align} 
where $a \in \R^n$ denotes the difference between the attack signal $\tilde{P}_d(a)$ and the actual signal $P_d$.

\subsubsection*{Attack Goal} The adversary uses these MCAs to manipulate~\eqref{eq:SecOPF}, so the new (deceived) generation output~$P_g^+(a) \in \Omega(\tilde{P}_d(a))$ increases the power flows on a set of target lines~$L \subset E$. Therefore, the attack goal is denoted as,
\begin{align} \label{eq:AttackGoal}
\Con := \{(l,\tau_l): [F]_l(\Pi_g P_g^+(a) - P_d) \geq (1+\tau_l) P_{f,l}(0)\}.
\end{align}
where $\Con \subset E \times \R_{>0}$ is the attack goal, $[F]_l$ is the $l$th row of the generation shifting matrix, $P_{f,l}(0)$ denotes the power flow on line~$l \in L$ before the MCA, and $\tau_l \in \R_{>0}$ quantifies the flow increase on $l \in L$. We choose this flow increase $\tau_l$ with semantics, including the flow increase that congests a transmission line or trips the line's protection.

\subsubsection*{Attack Constraints} MCAs are constrained based on the path they take on EMS. If the attack takes path 1 (see Fig.~\ref{fig:AttackPaths}), the MCA gets through state estimation and its data screening method. If the attack takes path 2 (see Fig.~\ref{fig:AttackPaths}), the MCA must take any value that deceives the operator. In any case, we can model this constraint as
\begin{align} \label{eq:AttackConst}
a \in [-\bar{a},\bar{a}].
\end{align}
In the above, $\bar{a} \in \R_{\geq 0}^n$ is the vector of max values allowed for the attack signal. We can use this vector to design different attack scenarios.
\begin{remark}
The constraint for path 1 can take a form that explicitly describes the condition under which measurement attacks get trough state estimation and its data screening methods. These methods, however, are not used during real-time and contingency dispatch (Path 2).
\end{remark}

MCAs are also constrained by defense at substations. If the grid's operator deploys defense at substation~$s_k$, the adversary cannot corrupt its measurements. We model this constraint as
\begin{align} \label{eq:MeasCorrup}
a_i \in \delta_{s_k}[-\bar{a}_i,\bar{a}_i],~\forall i \in V_{s_k},~ \forall s_k \in S,~\delta_{s_k} \in \{0,1\}.
\end{align} 
where $\delta_{s_k} = 1$ if measurements at substation $s_k$ are corruptible and $\delta_{s_k} = 0$ if not. The vector $\delta = [\delta_{s_1},\delta_{s_2},\hdots,\delta_{s_{n_s}}]^\top$ describes target and safe substations during MCAs.  Using~$\delta$, we can identify the set of target/attacked substations as follows 
\begin{align*}
\WC := \{s_k \in S~:~\delta_{s_k} = 1\} \in 2^S.
\end{align*}
Note that we can also use~\eqref{eq:MeasCorrup} to model the desire (for the adversary) to attack substation~$s_k$.

Finally, MCAs are constrained by the adversary's resources. If the adversary has limited resources, (s)he can only attack (hack) a limited number of substations. We model this constraint as
\begin{align} \label{eq:MaxAttackedSubs}
||\delta||_0 \leq \kappa,~\kappa \in \{1,2,\hdots,n_s\}.
\end{align}
In the worst case scenario for the operator, the adversary minimizes $\kappa$.

\subsubsection*{Types of Coordinated MCAs} Since the power grid is built with redundant measurements, attacking measurements in a single substation may not induce any consequence. In other words, effective MCAs are usually launched as a coordinated effort, which consists of temporally and spatially correlated events. Given the attack goal $\Con$, we classify coordinated MCAs as strongly and weakly correlated. \textit{Strongly Correlated MCAs}, denoted as $\SC \in 2^S$, achieve~$\Con$ by attacking the least number of substations. Strongly correlated MCAs describe attacks with minimum resources and allow us to predict attack consequences and derive defense implications. In Section IV, we will introduce a formal method to model and study strongly correlated MCAs. On the other hand, \textit{Weakly Correlated MCAs}, denoted as $\WC \in 2^S$, achieve~$\Con$ by attacking more substations than needed. Adversaries execute weakly correlated MCAs to probe defense at substations.
\vspace{-\dist}

%% file: CIG.tex
\section{Correlation Index Generator}
In this section, we describe the working principles of the Correlation Index Generator (see Fig.~\ref{fig:CIG}) and its components, namely the Induction Engine and the Deduction Engine.
\vspace{-\dist}
\subsection{Induction Engine}
Suppose the E-blocks detected an MCA~$\WC \in 2^S$ that is not in CKB and has corrupted measurements $\tilde{P}_d(a)$. The \textit{induction engine} computes the new (deceived) generation output $P_g^+(a)$ by solving $\Omega(P_d+a)=: \Omega(\tilde{P}_d(a))$, \ie 
\begin{align*}
\Omega(\tilde{P}_d(a)):~\min_{P_g} \quad & \frac{1}{2} P_g^\top C_2 P_g + c_1^\top P_g +c_0,  \\
					 \text{s.t.} \quad & \mathbf{1}^\top P_g - \mathbf{1}^\top(P_d+a) = 0, \\
                  & P_g \in [\mathbf{0},\bar{P}_g], \\
                  & F(\Pi_g P_g - (P_d+a)) \in [-\bar{P}_f,\bar{P}_f].  
\end{align*}

Then, using $P_g^+(a) \in \Omega(\tilde{P}_d(a))$, the \textit{induction engine} determines the set of attack consequences, \ie the set $\Con$. As shown in \eqref{eq:AttackGoal}, the set of consequences~$\Con$ depends on $\tau_l$ and $P_d$. The flow increase $\tau_l$ is chosen with semantics and the real consumption~$P_d$ is obtained as follows.
\begin{align*}
P_{d,i}:= \begin{cases}
\tilde{P}_{d,i},\quad &\text{if } i \not \in V_{s_k},~\forall s_k \in \WC, \\
P_{d,i}^{\text{pre}},\quad &\text{otherwise},
\end{cases}
\end{align*}
where $P_{d,i}^{\text{pre}}$ is a (conservative) estimated consumption or a redundant measurement.
\vspace{-\dist}
\subsection{Deduction Engine}
Given the set of consequences inflicted~$\Con$, the \textit{deduction engine} computes strongly correlated MCAs that reach $\Con$ using the following bilevel mix-integer optimization program:
\vspace{-\dist}
\begin{subequations} \label{eq:AttackTemplate}
\begin{align}
\min_{P_g^{+},a,\kappa,\delta} \quad & \kappa,\\
\text{s.t.} \quad &\text{equations}~\eqref{eq:MeasAttacks}-\eqref{eq:MaxAttackedSubs}, \\
                  & P_g^+ \in \Omega(\tilde{P}_d(a)).
\end{align}
\end{subequations}

In our previous work~\cite{MoyaCIs}, we derived a method that addresses the mathematical challenges of~\eqref{eq:AttackTemplate} and computes strongly correlated MCAs. The method first computes the security index, which corresponds to the optimal solution~$\kappa^*$. This security index describes the minimum number of substations the adversary must attack to reach~$\Con$. Then, the method determines the target and safe substations during the MCA from the optimal solution $\delta^*$. Since $\delta^*$ is not necessarily unique, we proposed in~\cite{MoyaCIs} an algorithm to determine all feasible solutions $\delta^*$ such that $||\delta^*||_0 = \kappa^*$. All these $\delta^*$ correspond to strongly correlated MCAs associated with the attack goal~$\Con$.

We use a set-theoretic approach to describe all these strongly correlated MCAs, which we define as Correlation Indices.
\begin{definition}
Let $\delta^*$ denote a feasible solution of~\eqref{eq:AttackTemplate} associated with~$\Con$ such that $||\delta^*||_0 = \kappa^*$. A \textit{Correlation Index}~(CI), denoted as $\SC$, is a strongly correlated MCA that extracts target substations from $\delta^*$ as follows
\begin{align*}
S^*_{\alpha,j}:=\{s_k \in S~:~\delta^*_{s_k} \neq 0\} \in 2^S,
\end{align*}
and inflicts the consequences described by~$\Con$.
\end{definition}
The set of all CIs associated with the inflicted consequences~$\Con$ is given by $\mathcal{S}_{\alpha,j}^*:=\{\SC~:~\SC~\text{is a CI}\}$.

As a result, the CIG generates a \textit{CI-target tuple} $(\mathcal{S}_{\alpha,j}^*, \Con)$ --\ie the set of strongly correlated MCAs and the associated inflicted consequences-- and sends this CI-target tuple to the Correlation Knowledge-Base~(CKB). 
\vspace{-\dist}

%% file: CKB.tex
\section{Correlation Knowledge-Base}
In this section, we describe the working principles of the Correlation Knowledge-Base~(CKB) (see Fig.~\ref{fig:CKB}) using a set-theoretic approach. The CKB has a Scanning Engine and a Reasoning Engine.
\vspace{-\dist}
\subsection{Scanning Engine}
Suppose the E-blocks detected a (possibly weakly correlated) MCA~$\WC$. The \textit{Scanning Engine} verifies if~$\WC$ is an existing MCA, \ie if $\WC \in \text{CKB}$. The MCA~$\WC$ is an existing MCA if
\begin{enumerate}
\item The MCA is a CI (or strongly correlated MCA), \ie $\WC \in \mathcal{S}^*_{\alpha,j}$ for some $\mathcal{S}^*_{\alpha,j} \in \text{CKB}$.
\item The MCA is a weakly correlated MCA but a superset of at least one CI, \ie $\exists \SC \subset \WC$ such that $\SC \in \text{CKB}$.
\item The MCA is uncorrelated, is a subset of at least one CI, \ie $\exists \SC \supset \WC$ such that $\SC \in \text{CKB}$, and has less cardinality than all CIs in CKB, \ie $|\WC| < |\SC|$ for all $\SC \in \text{CKB}$.
\end{enumerate}

If $\WC$ is an existing MCA, then CKB uses the reasoning engine to identify physical targets and derive defense strategies. Otherwise, CKB calls the CIG to analyze $\WC$.
\vspace{-\dist}
\subsection{Reasoning Engine}
The \textit{reasoning engine} identifies physical targets and derives defense strategies for the detected MCA~$\WC$. Technically, the reasoning engine is an R-block (see Fig.~\ref{fig:IDSCIGCKB}) and can work also with CIG to derive defense strategies. 

To identify physical targets associated with~$\WC$, we proceed as follows. 
\begin{enumerate}
\item  If the MCA~$\WC$ is a CI, then the physical targets are described by the set of inflicted consequences~$\Con$.
\item  If the MCA~$\WC$ is a weakly correlated MCA that contains a set of $q \geq 2$ CIs, \ie the set \[\mathcal{S}_{\text{CI}}:=\{ S^*_{\alpha,j}~:~j=1,\hdots,q~\text{and}~S^*_{\alpha,j} \subset \WC \},\] then the physical targets are given by the union of the inflicted consequences associated with each CI, \ie $\cup_{j=1}^q \Con$ where $(S^*_{\alpha,j},\Con)$ is a CI-tuple of an existing MCA. 
\end{enumerate}

To derive defense strategies against~$\WC$, we proceed as follows. 
\begin{enumerate}
\item If the MCA~$\WC$ is a CI, then the best defense strategy is to defend any substation. 
\end{enumerate}

This defense will render the attack ineffective, which we justify next.
\begin{proposition} \label{prop:1}
(Defense against strongly correlated MCAs) Let $\WC$ denote a strongly correlated MCA. If the operator protects measurements at any substation substation~$s^*_k$ such that~$s^*_k \in \WC$, the attack~$S_{\alpha,j} \setminus \{s^*_k\}$ becomes ineffective.
\end{proposition}
\begin{proof}
See Appendix.
\end{proof}

\begin{enumerate}
  \setcounter{enumi}{1}
  \item If the MCA~$\WC$ is a weakly correlated MCA that contains the set of CIs $\mathcal{S}_{\text{CI}}$, then we may have one of the following cases.
\end{enumerate}

\noindent \textbf{Case I:} If $\cap_{j=1}^q S^*_{\alpha,j} \neq \emptyset$, then the best defense strategy is to protect measurements at substation~$s^*_{k}$ that satisfies~$s^*_{k} \in \cap_{j=1}^q S^*_{\alpha,j}$, which we justify next.
\begin{proposition}
(Defense against a set of strongly correlated MCAs with non-empty intersection) Let $\WC$ denote a weakly correlated MCA that contains the set of CIs $\mathcal{S}_{\text{CI}}$. Suppose these CIs satisfy $\cap_{j=1}^q S^*_{\alpha,j} \neq \emptyset$. If the operator protects measurements at a substation~$s^*_{k}$ such that~$s^*_{k} \in \cap_{j=1}^q S^*_{\alpha,j}$, the attack~$S_{\alpha,j} \setminus \{s^*_k\}$ becomes ineffective. 
\end{proposition}
\begin{proof}
Follows from Proposition~\ref{prop:1}.
\end{proof}

\noindent \textbf{Case II:} If $\cap_{j=1}^q S^*_{\alpha,j} = \emptyset$, then the best strategy is to defend all CIs individually, which we justified using Proposition~\ref{prop:1}.

\noindent \textbf{Case III:} Finally, there is an intermediate case in which only some CIs have a non-empty intersection. For this case, a combination of the defense strategies described for Case I and II should be implemented.
\vspace{-\dist}

%% file: Experiments.tex
\section{Numerical Experiments}
In this section, we use numerical experiments to validate our proposed framework. In particular, we compute the false alarm rates for CIG and CKB under different attack scenarios.
\subsection{Experimental Setup}
We describe the experimental environment, the IDS benchmark systems, and the evaluation metric next.  
\subsubsection{Environment}
We model a power grid with $n_s = 6$ substations using the New England 39-bus system illustrated in Fig.~\ref{fig:NewEngland39}. We model the dispatch application using the DC Optimal Power Flow tool from MatPower~\cite{zimmerman2011matpower}. The data used for the power grid and dispatch application corresponds to Matpower base-case data.
\begin{figure}[t]
\centering
\includegraphics[width= .40\textwidth, height = 2.25 in]{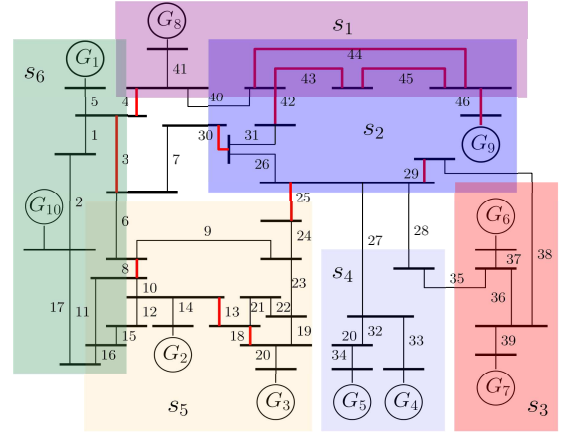}
\caption{New England 39-bus system.}
\label{fig:NewEngland39}
\end{figure}

In our experiments, we used the adversarial environment introduced in~\cite{cardenas2006framework}. This adversarial environment is characterized by a nominal attack rate (or attack intensity) $p_0 \in (0,1)$, which E-blocks estimate as $\hat{p}_0$. 

We model MCAs using a random approach, that is, we selected the corrupted measurements~$\tilde{P}_d(a)$ and target substations~$\WC$ uniformly at random. In particular, $\tilde{P}_d(a)$ was chosen uniformly from the interval $[P_d-\bar{a},P_d+\bar{a}]$ where $\bar{a} = 0.1P_d$. This random approach allowed us to model attack events that are a threat and attack events that are not. 
\subsubsection{Intrusion Detection Systems}
We model E-blocks (or IDS's detector) with the following characteristics. The E-blocks have a detection rate $p_D \in (0,1)$ and a false alarm rate $p_{FA} \in (0,1)$. In our experiments, we selected the values of $p_D = 0.9$, $p_{FA} = 0.1$. The adversary attempts to manipulate the E-blocks' $p_D$, $p_{FA}$, and $\hat{p}_0$ by using the following parameters: 
\begin{itemize}
\item $\delta$:~the maximum deviation under $\hat{p}_0$.
\item $\beta$:~the maximum probability to launch a zero-day (\ie undetectable) attack.
\item $\alpha$:~the maximum probability to intentionally trigger a false alarm. 
\end{itemize} 
In the simulation, we selected the values $\delta = 0.1\hat{p}_0$, $\beta = 0.2$, $\alpha = 0.1$, and $\hat{p}_0 \in \{0.25,0.1,0.05\}$.

We model two benchmark IDSs, a simple IDS (IDS-1) and a Bayesian IDS (IDS-2). IDS-1 has the following working principle. If the E-blocks trigger an alarm, IDS-1 will label the event as an intrusion. IDS-2, on the other hand, has the following working principle. An event is labeled as an intrusion based on $\mathbb{P}(\text{Intrusion}|\text{Alarm})$, \ie the probability of intrusion given that an alarm has been triggered. This probability is computed as follows
\begin{align*} 
\mathbb{P}(I|A) = \frac{\mathbb{P}(A|I) \mathbb{P}(I)}{\mathbb{P}(A|I) \mathbb{P}(I)+\mathbb{P}(A|\neg I) \mathbb{P}(\neg I)},
\end{align*}
where $A$ denotes the alarm and $I$ intrusion. Since $ \mathbb{P}(I) = \hat{p}_0$, $ \mathbb{P}(A|I) = p_D$, $ \mathbb{P}(\neg I) = 1-\hat{p}_0$, and $ \mathbb{P}(A|\neg I) = p_{FA}$; we write $\mathbb{P}(I|A)$ as
\begin{align} \label{eq:PPV2}
\mathbb{P}(I|A) = \frac{p_D \hat{p}_0}{(p_D - p_{FA})\hat{p}_0 + p_{FA}},
\end{align}
which is also known as the \textit{Bayesian detection rate}~\cite{cardenas2006framework}.

To model CKB and CIG, we proceed as follows. For CKB, we computed CI-tuples for each experiment using CVX and Gurobi, packages for specifying and solving convex and mix-integer programs~\cite{cvx}. CIG detects possible threats based on deviation from the pseudo-measurements ~$P_d^{\text{pre}}$, which are generated from a uniform distribution in $[0.9P_d,1.1P_d]$. We assume no redundant measurements are available for CIG to replace the corrupted measurements. Nevertheless, if they are available, the false alarms (for CIG) will tend to 0.

CKB and CIG will label an incoming MCA as a threat, if the attack can increase the flow~$\tau_l = 15\%$ in any of the following target lines $L = \{3,4,13,18,25,29,30,42,43,44,45,46\}$ (see Fig.~\ref{fig:NewEngland39}). This requires for CKB to have CI-tuples for each line $l \in L$..

\subsubsection{Metrics}
The performance of the benchmark IDSs and the proposed framework is measured by the false negative rate $\text{FNR} := \text{FN}/ (\text{TP}+\text{FN})$,
where FN denotes the false negatives (\ie failure of generating an alarm) and TP the true positives (\ie success of generating an alarm correctly), and the false positive rate $\text{FPR}:= \text{FP}/(\text{TN}+\text{FP})$, where FP denotes the false positives (\ie generating a false alarm) and TN the true negatives (\ie stay silent when there is no event). 

We further define these metrics for intrusions that are not a threat (\ie \textit{ineffective attacks}) and for intrusions that are a \emph{threat} (denoted as $\text{FNR}_t$ and $\text{FPR}_t$). Since IDS-1 and IDS-2 are not capable of estimating attack consequences and determining possible threats, we compute $\text{FNR}_t$ and $\text{FPR}_t$ only for the proposed framework. All metrics are evaluated through a large sample of events using the pseudo-code algorithm described in Appendix~B.
\subsection{Experimental Results} 
\subsubsection*{Experiment I} \textit{False Alarm Rates}. In this experiment, we computed the FNR and FPR for IDS-1, IDS-2, and CKB/CIG. We used the pseudo-code to simulate $M = 10^2$ experiments of $N= 10^3$ attack/normal events. Fig.~\ref{fig:FNR} shows the FNRs and Fig.~\ref{fig:FPR} the FPRs (using box plots) for the attack rates $\hat{p}_0 \in \{0.25,0.05\}$. 

For the FNR case, the results show that for both $\hat{p}_0 = 0.25$ and $\hat{p}_0 = 0.05$, CKB/CIG outperforms IDS-2 but not IDS-1. This is because CKB and CIG label an event as an intrusion if and only if the event threatens the power grid. As a result, ineffective attacks are not labeled as intrusions, which increases the number of false negatives. If instead of computing the FNR for intrusions, we compute the FNR for threats (\ie $\text{FNR}_t$), then we will see how CKB and CIG outperform IDSs with no contextual information, which we describe in Experiment II.

For the FPR case, the results show that for $\hat{p}_0 = 0.25$, CKB/CIG performs worse than for IDS-1 and IDS-2. In a more friendly environment, \ie when $\hat{p}_0 = 0.05$, CKB/CIG outperforms IDS-1 but not IDS-2. This is because (i) the fast screening of CKB increases the number of false positives in a less friendly environment and (ii) CKB is sensitive to the number of critical targets (\ie the cardinality of $L$), which we describe in Experiment III.
 
\begin{figure}[h!]
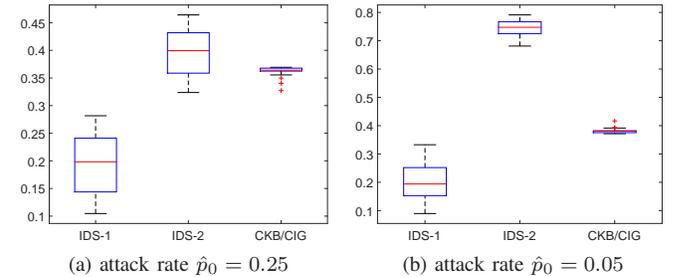
%
    \centering
    \subfloat[attack rate $\hat{p}_0 = 0.25$]{{\includegraphics[width= .225\textwidth, height = 1.25 in]{fig/FNR_1q4.eps} }}%
    ~
    \subfloat[attack rate $\hat{p}_0 = 0.05$]{{\includegraphics[width= .225\textwidth, height = 1.25 in]{fig/FNR_1q20.eps} }}%
    \caption{FNR for IDS-1, IDS-2, and CKB/CIG}%
    \label{fig:FNR}%
\end{figure}
\begin{figure}[h!]
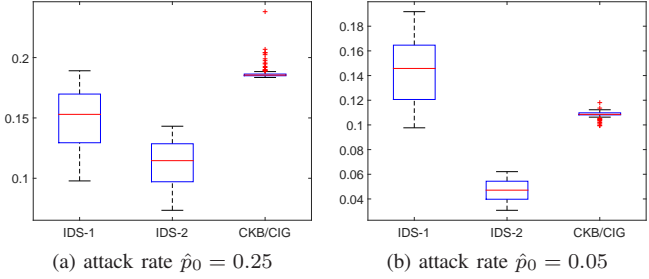
%
    \centering
    \subfloat[attack rate $\hat{p}_0 = 0.25$]{{\includegraphics[width= .225\textwidth, height = 1.25 in]{fig/FPR_1q4.eps} }}%
    ~
    \subfloat[attack rate $\hat{p}_0 = 0.05$]{{\includegraphics[width= .225\textwidth, height = 1.25 in]{fig/FPR_1q20.eps} }}%
    \caption{FPR for IDS-1, IDS-2, and CKB/CIG}%
    \label{fig:FPR}%
\end{figure}
\subsubsection*{Experiment II} \textit{Threat Analysis}. 
In this experiment, we computed the FNR for threats~(\ie $\text{FNR}_t$). A false negative occurs if the random MCA was a threat for the power grid but CKB and CIG determined that it was not a threat. Fig.~\ref{fig:FNR_Threat} shows the FNRs for the attack rates $\hat{p}_0 \in \{0.25,0.1,0.05\}$. As expected, the contextual information used by CKB and CIG considerably decreases the FNR for threats.

\begin{figure}[t]
\centering
\includegraphics[width= .30\textwidth, height = 1.5 in]{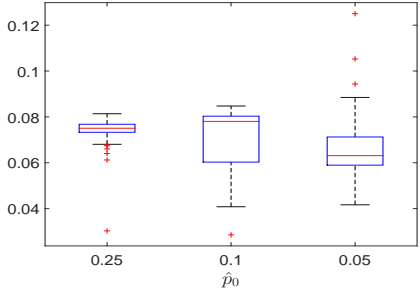}
\caption{$\text{FNR}_t$-Threat analysis.}
\label{fig:FNR_Threat}
\end{figure}
\subsubsection*{Experiment III} \textit{Sensitivity Analysis}. In this experiment, we studied the sensitivity of $\text{FPR}_t$ to the cardinality of $L$. Table~\ref{table:sensitivity} shows that the average $\text{FPR}_t$, denoted as $\mu(\text{FPR}_t)$, decreases as the number of critical/target lines decreases. This is because, as the number of critical lines decreases, the number of CI-tuples stored in CKB decreases too. As a result, the fast scanning feature of CKB will be less prone to false positives. 

Note that there is a trade-off between the number of critical targets selected and the maximum $\text{FPR}_t$ allowed, which should be adjusted based on risk assessment or experience. A different solution would be to always use CIG. This, however, will greatly increase the runtime of our proposed framework.

{\renewcommand{\arraystretch}{1.2}
\begin{table}[!t]
\caption{Sensitivity of $\text{FPR}_t$.}
\label{table:sensitivity}
\centering
\begin{tabular}{c c c} 
\hline  
 $\tau_l$ & $L$ & $\mu(\text{FPR}_t)$ \\  
\hline   
10 \% & $\{3,4,13,18,25,30,42,44,45,46\}$ & 0.1806 \\ 
20 \% & $\{3,4,13,18,25,30,36,42,43,44,45,46\}$ & 0.1774 \\
20 \% & $\{3,4,25,30,45,46\}$ & 0.0625 \\
30 \% & $\{4,13,18,25,42,45,46\}$ & 0.0646 \\
30 \% & $\{25,46\}$ & 0.0351 \\
40 \% & $\{13\}$ & 0.0190 \\
 \hline
\end{tabular}
\end{table}}

%% file: PseudoCode.tex
\section{Pseudo-Code}
We use Algorithm~\ref{alg:1} to compute the FNR/$\text{FNR}_t$ and FPR/$\text{FPR}_t$ for CIG and CKB. Some remarks on Algorithm~\ref{alg:1} are the following. (i) The if-conditionals describe how the E-blocks, CKB, and CIG interact during normal/attack events. (ii) Algorithm~\ref{alg:1} describes attacks at the grid level, that is, either the grid is under attack~$\WC$ or not. We remark, however, that it can be easily adapted to model attacks at the substation level, that is, individual substations are under attack or not. (iii)~Finally, by making the appropriate changes, Algorithm~\ref{alg:1} can compute the FNR and FPR for IDS-1 and IDS-2.
\begin{algorithm}[h!]
\footnotesize
\caption{Deriving FNR and FPR for CKB/CIG} \label{alg:1}
\begin{algorithmic}[1]
\State $(\text{FNR},\text{FPR}) \gets$ Rates( )
\Procedure{Rates}{ }  
\For{$k = 1$ to $M$}  \Comment{M: number of experiments}
\For{$i = 1$ to $N$}  \Comment{N: number of attack/normal events}
\State Select $p_1 \in [\hat{p}_0-\delta,\hat{p}_0+\delta]$
\State Select $p_2 \in [0,\beta]$
\State Select $p_3 \in [0,\alpha]$
\State $\text{Attack}(i) \gets \text{Bernoulli}(p_1)$   \Comment{Initiate attack/normal event}
\If{$\text{Attack}(i) = 1$}   \Comment{Attack event}
\State $[\WC, \tilde{P}_d(a)] \gets$ RandomMCA
\State ZD $\gets$ Bernoulli($p_2$) \Comment{ZD: zero-day attack}
\State Alarm $\gets$ Bernoulli($\min \{1-\text{ZD},p_D\}$)
\If{$\text{Alarm} = 1$}
\State [ Threat($i$),isMCA ] $\gets$ CKB($\WC$)
\If{isMCA  = No}
\State Threat($i$) $\gets$ CIG($\WC,\tilde{P}_d(a)$)
\EndIf
\Else          \Comment{Zero-day attack case}        
\State Threat($i$) $\gets$ CIG($\emptyset,\tilde{P}_d(a)$)
\EndIf
\Else  \Comment{Normal event}
\State FA $\gets$ Bernoulli($p_3$) \Comment{FA: false alarm}
\State Alarm $\gets$ Bernoulli($\max \{\text{FA},p_{FA}\}$)
\If{Alarm = 0}
\State Threat($i$) $\gets$ CIG($\emptyset,\tilde{P}_d$)
\Else                           \Comment{False alarm case}
\State $[\WC, \tilde{P}_d] \gets$ RandomMCA  
\State [ Threat($i$),isMCA ] $\gets$ CKB($\WC$)
\If{isMCA  = No}
\State Threat($i$) $\gets$ CIG($\WC,\tilde{P}_d$)
\EndIf
\EndIf
\EndIf
\EndFor
\State FN, FP, TN, TP $\gets$ from Attack and Threat
\State Compute FNR($k$) and FPR($k$) 
\EndFor
\State \textbf{return}(FNR,FPR)
\EndProcedure
\end{algorithmic}
\end{algorithm}